\documentclass[11pt]{article}

\usepackage[margin=1in]{geometry}
\usepackage{times}
\usepackage{graphicx}
\usepackage{amsmath, amssymb, amsthm}
\usepackage{geometry}
\usepackage{enumitem}

\usepackage{booktabs}
\usepackage{xurl}
\usepackage[hidelinks]{hyperref}
\usepackage{natbib}
\usepackage{xcolor}
\usepackage{subcaption}
\usepackage{multirow}
\usepackage{placeins}
\usepackage[utf8]{inputenc}

\newtheorem{definition}{Definition}
\newtheorem{theorem}{Theorem}
\newtheorem{lemma}{Lemma}
\newtheorem{remark}{Remark}

\title{Agent-based Modeling: Equilibrium, Echo Chambers, and Efficiency in Hybrid Coevolutionary Opinion Games}

\author{Ming-Zhi Jiang\thanks{Department of Information Management \& Finance,
National Yang Ming Chiao Tung University, Taiwan. s112701018.mg12@nycu.edu.tw} \and
An-Tzu Teng\thanks{Institute of Information Management, National Yang Ming Chiao Tung University, Taiwan. a0783403.c@nycu.edu.tw} \and
Jun-En Liu\thanks{Department of Information Management \& Finance,
National Yang Ming Chiao Tung University, Taiwan. neil.mg12@nycu.edu.tw} \and
Po-An Chen\thanks{Institute of Information Management, National Yang Ming Chiao Tung University, Taiwan. poanchen@nycu.edu.tw} \and
Yung-Ming Li\thanks{Institute of Information Management, National Yang Ming Chiao Tung University, Taiwan. yml@nycu.edu.tw}
}

\date{}

\begin{document}

\maketitle

\begin{abstract}
Online discussion of political and gender-related issues is often heated, and when opinions
in a network draw closer, the convergence is readily taken as genuine consensus. Whether it
carries a cost is a question existing methods cannot answer: coevolutionary opinion formation
games measure the Price of Anarchy (PoA) of agents that update by numerical rules, while
simulations with large language model (LLM) agents report only descriptive indices.

We introduce the Hybrid Coevolutionary Opinion Game (H-COG), in which analytical and
LLM-driven agents share one network, choose their neighbors by opinion similarity in every
round, and hold stances drawn from real Reddit comments on gun control and abortion. To our
knowledge, H-COG is the first framework to place Friedkin--Johnsen best-response agents and
LLM agents in one coevolutionary game and to measure the social cost and PoA of LLM-driven
populations. We prove that on any fixed network, given the LLM agents' opinions, the
analytical agents' opinion stage has a unique equilibrium and the social optimum has a
closed form, and that the convergence guarantee of Chen et al. for optimistic gradient
ascent carries over to H-COG.

All runs converge structurally. LLM-driven populations are less polarized yet have about
five times the PoA of analytical ones; half of the gap comes from agents being pulled away
from their own prior positions, a distance we prove must carry a cost whenever expressed
opinions are more concentrated than intrinsic ones. Echo chambers form under every composition and grow out of the rewiring rule
rather than the initial topology. Opinions in these populations draw closer largely because
agents give up their own positions.
\end{abstract}
\medskip\noindent
\textbf{Keywords:} Opinion dynamics, Echo chambers, Price of Anarchy, Language models, Agent-based modeling

\section{Introduction}
% ===== 1 Introduction =====
Online platforms have turned the public square into a coevolving system: what a user reads
shapes what they think, and what they think reshapes whom they follow. Echo chambers are the
visible result. A population that mostly meets what already agrees with it polarizes further,
which is why the phenomenon has become a societal concern.
Explaining it means tracking opinions and ties together, not either alone.

Two lines of work approach this, and each gives up something the other keeps. Analytical
opinion models, among them DeGroot, Friedkin--Johnsen (FJ), bounded confidence and the
coevolutionary opinion games~\citep{degroot1974,friedkin1990,fotakis2016opinion,bhawalkar2013},
supply explicit update rules and formal tools for equilibrium, social cost and the Price of
Anarchy (PoA). What they give up
is language: each exchange becomes a numerical update, and the semantic content through which
people actually persuade one another drops out. LLM-based agent simulations keep that content,
letting agents read, reason and remember, and they reproduce echo chambers and
polarization~\citep{wang2025,chuang2024simulating,gu2025,piao2025emergence}. What they give up
is the yardstick. These studies aim at behavioral realism, and the simulators that mix LLM agents
with rule-based ones do so without a game or a rewiring
network~\citep{mou2024hisim,yao2025fde,chen2026agorasim}. No framework yet measures social cost
and PoA where language-driven and cost-driven agents coevolve with one network.

H-COG closes that gap by putting both kinds of agent in one network that rewires as opinions
change. A Type-C agent minimizes a numerical disagreement cost, while a Type-L agent
reads what its neighbors wrote and answers in language. The parameter \(\alpha\) fixes how many of each the population holds, both are read on
the same continuous scale, and every agent re-chooses its neighbors by closeness of position.
Sweeping \(\alpha\) therefore compares purely numerical, purely LLM-driven and mixed populations
under one set of conditions. Efficiency is measured by social cost and a topology-conditioned
PoA, with the observed cost and the planner optimum computed on the same realized network after
each rewiring.

The simulations rest on 5,199 \texttt{r/politics} comments about gun control and abortion. A
fine-tuned RoBERTa regressor maps each comment to a continuous stance in \([-1,+1]\), and
stratified samples from the resulting distributions set the agents' intrinsic opinions. The two
topics differ in how their stances are spread, which lets us ask whether the effect of
composition survives a change in the starting landscape. We sweep nine mixing levels of
\(\alpha\) from 0 to 1 across three initial topologies.

Three questions organize the study: how large an efficiency gap composition creates and
which cost term produces it (Section~\ref{sec:res-rq1}), how composition shapes polarization
and echo chambers (Section~\ref{sec:res-rq2}), and what the initial topology and the rewiring
rule each contribute to clustering (Section~\ref{sec:res-topology}).

\subsection{Contributions}
This study makes five contributions:
\begin{enumerate}
    \item \textbf{Hybrid coevolutionary framework.} We introduce H-COG, which integrates
    cost-minimizing Type-C agents and LLM-driven Type-L agents in a common network whose opinions
    and connections coevolve.
    \item \textbf{Equilibrium analysis.} We formalize H-COG as a multi-player general-sum Markov
    game, prove that its opinion stage has a unique best response, an invariant opinion
    interval, a unique Nash equilibrium on every fixed graph given the Type-L opinions, and a closed-form planner optimum,
    show that $K$-NN rewiring is the high-intensity limit of a stochastic kernel, bound the
    conformity cost below by the gap between the spreads of intrinsic and expressed
    opinions, and combine these
    results with a convergence theorem for Markov coevolutionary opinion games.
    \item \textbf{Empirically grounded initialization.} We construct a two-stage pipeline that
    extracts real-world stance distributions from Reddit using semantic clustering and a
    fine-tuned RoBERTa regressor, then initializes agents through stratified sampling.
    \item \textbf{Common efficiency evaluation.} We formulate social cost and
    topology-conditioned PoA on a shared opinion scale, and decompose social cost into disagreement
    and conformity components to identify the source of efficiency differences.
    \item \textbf{Systematic population comparison.} We compare nine Type-C/Type-L mixing levels
    across two topics and three initial topologies, evaluating efficiency, polarization,
    echo-chamber structure, and convergence behavior.
\end{enumerate}

\section{Related Work}
\label{sec:related}

\subsection{Opinion Dynamics and Coevolutionary Games}

Opinion dynamics begins with agents that hold a number and update it by averaging their
neighbors. DeGroot formalized that process and showed that under mild conditions repeated
averaging drives the whole network to consensus~\citep{degroot1974}. Friedkin and Johnsen
refined it by separating the belief an agent holds privately from the one it expresses, which
lets the system settle where disagreement survives instead of vanishing~\citep{friedkin1990}.

Bhawalkar et al.\ took the further step of letting the network move as well~\citep{bhawalkar2013}.
In their $K$-NN game each agent links to the $K$ agents whose expressed opinions lie closest to
its intrinsic opinion and picks an expressed opinion that trades social pressure against private
belief. That setup admits a Price of Anarchy (PoA) analysis. The PoA of the $K$-NN game is bounded
by a constant when agents weight their own opinion above their neighbors' ($\rho>1$) and is at
least $1/\rho^2$ below it, and a pure equilibrium need not exist; H-COG with $\alpha=1$ is this
game with heterogeneous stubbornness. On fixed undirected graphs the Friedkin--Johnsen game loses
at most a factor of $9/8$, and on directed graphs the loss is unbounded~\citep{bindel2015bad}.
Chen et al.\ bound the PoA for directed graphs in which each node influences the others no more
than it is influenced by them~\citep{chen2016}. \citet{chen:lu:lin:shi:hung} make $K$-NN
rewiring stochastic, cast the game as a general-sum Markov game, and show that optimistic
gradient ascent reaches an approximate Nash equilibrium or, otherwise, a bounded price of
anarchy;
Section~\ref{sec:equilibrium} builds on this result. Related work extends the analysis to
competitive opinion optimization~\citep{chen:lu:lin}, designs networks that minimize
polarization and disagreement at the Friedkin--Johnsen equilibrium~\citep{musco2018minimizing},
and proves convergence of Hegselmann--Krause
dynamics with restricted neighborhoods~\citep{fotakis2016opinion}. All of these models use
numerical update rules and leave out the reasoning, memory and meaning through which people
form and state an opinion.

\subsection{LLM-based Agent Simulation for Echo Chambers}

Large language models let an agent read text, reason about it and answer in kind. Park et al.\
showed that such agents produce believable social behavior in interactive
settings~\citep{park2023}.

Wang et al.\ placed LLM agents on fixed network
structures, with a recommendation algorithm selecting whom each agent hears, and found that the
agents track observed polarization more closely than Friedkin--Johnsen or bounded-confidence
models do~\citep{wang2025}. Chuang et al.\ showed that networks of LLM agents drift toward
consensus unless confirmation bias is induced~\citep{chuang2024simulating}. Gu et al.\ let the
language model decide rewiring as well, first on real Twitter data~\citep{gu2025} and then with
paired rewiring validated on three real datasets~\citep{gu2026language}, and Piao et al.\ observed
polarization emerging in networks that thousands of LLM agents form themselves
\citep{piao2025emergence}.

A second line mixes LLM agents with rule-based ones. HiSim pairs LLM core users with ordinary
users governed by bounded-confidence and related models, initialized from real tweets
\citep{mou2024hisim}; FDE-LLM couples LLM opinion leaders with dynamics equations on real Weibo
data~\citep{yao2025fde}; and AgoraSim runs populations that mix LLM and classical agents at
controlled ratios~\citep{chen2026agorasim}. None of these simulators rewires the network or
places the two types in a game. Qasmi et al.\ place LLM agents in a non-cooperative game of
opinion polarization without an equilibrium or efficiency analysis~\citep{qasmi2025competing},
and welfare measures for LLM agent societies exist outside opinion formation
\citep{piatti2024cooperate}.

These studies report descriptive indices such as polarization and modularity. None measures
what the resulting equilibria cost, and none places LLM agents and Friedkin--Johnsen
best-response agents in one coevolving game.

\section[Preliminaries: Formalization of the Hybrid Markov Coevolutionary Opinion Formation Game]{Preliminaries: Formalization of the Hybrid Markov\texorpdfstring{\\}{ } Coevolutionary Opinion Formation Game}
\label{sec:prelim}
\subsection{Agent Population and Opinion Variables}
\label{sec:framework-agents}

Let $\mathcal{N}=\{1,\ldots,N\}$ be a set of $N$ agents. Each agent $i$ has an
\textbf{intrinsic opinion} $s_i\in[-1,+1]$, representing a fixed prior stance,
and an \textbf{expressed opinion} $z_i^{(t)}\in[-1,+1]$, which evolves as the
agent responds to its neighbors.

The mixing parameter $\alpha\in[0,1]$ partitions the population into two types.
\textbf{Type-C} agents (fraction $\alpha$) update their expressed opinions by
minimizing a local disagreement cost derived from the Friedkin--Johnsen
model~\citep{friedkin1990}. \textbf{Type-L} agents (fraction $1-\alpha$)
generate a natural-language opinion $o_i^{(t)}\in\mathcal{V}^*$, where
$\mathcal{V}$ is a finite token vocabulary. To align text-based interaction with
numerical game-theoretic optimization, we define a continuous semantic mapping
\[
e:\mathcal{V}^*\rightarrow[-1,+1]
\]
that places this opinion on the common numerical scale,
$z_i^{(t)}=e(o_i^{(t)})$. The cases $\alpha=1$ and $\alpha=0$ give
the purely numerical and purely LLM-driven baselines, respectively.

We write $\mathcal{N}_C$ and $\mathcal{N}_L$ for the sets of Type-C and Type-L
agents. Every agent carries a stubbornness coefficient $\rho_i>0$. For a Type-C
agent it enters the update rule, and for a Type-L agent it enters only the
social cost. Opinions are represented at finite precision: we write
$\mathcal{Z}\subset[-1,+1]$ for the finite set of representable opinions and
$\mathcal{V}^{\le M}$ for texts of at most $M$ tokens.

\subsection{State Space and Opinion Updates}
\label{sec:framework-updates}

The global state at step $t$ is
\[
S^{(t)}=\left(G^{(t)},\mathbf{z}^{(t)},\mathbf{X}^{(t)}\right),
\]
where $G^{(t)}=(S_1^{(t)},\ldots,S_N^{(t)})$ is the directed network,
$S_i^{(t)}\subset\mathcal{N}\setminus\{i\}$ is agent $i$'s out-neighborhood
with $|S_i^{(t)}|=K$ for $t\ge1$ ($S_i^{(0)}$ is $i$'s neighbor set in the initial graph, on
which Eq.~\eqref{eq:type-c-update} is applied with $K$ replaced by $|S_i^{(0)}|$, and with $z_i^{(1)}=s_i$ if $S_i^{(0)}=\emptyset$), $\mathbf{z}^{(t)}$ is the opinion profile from which
$G^{(t)}$ was built, and $\mathbf{X}^{(t)}=(\mathbf{x}_i^{(t)})_{i\in\mathcal{N}_L}$
collects the structured text records of the Type-L agents that produced
$\mathbf{z}^{(t)}$.

\subsubsection*{Type-C analytical agents}
A Type-C agent chooses an expressed opinion from
$\mathcal{A}_i=\mathcal{Z}$. Holding its neighbors' latest opinions fixed, its
myopic stage-cost minimizer is
\begin{equation}
\label{eq:type-c-update}
z_i^{(t+1)}
=\arg\min_{z\in[-1,+1]}C_i\!\left(G^{(t)},(z,\mathbf{z}_{-i}^{(t)})\right)
=\frac{\sum_{j\in S_i^{(t)}}z_j^{(t)}+\rho_iKs_i}
       {K(1+\rho_i)}.
\end{equation}
This is the unique minimizer of the stage cost, and it stays in $[-1,+1]$
(Lemmas~\ref{lem:br} and~\ref{lem:inv}). In the Markov-game representation, a
more general Type-C policy $\pi_i^C(\,\cdot\mid S^{(t)})$ may instead optimize
expected discounted costs; the simulations use the myopic update in
Eq.~\eqref{eq:type-c-update}.

\subsubsection*{Type-L generative agents}
A Type-L action is a structured sequence
$\mathbf{x}_i^{(t+1)}\in\mathcal{V}^{\le M}$ containing
\texttt{reasoning}, \texttt{opinion}, and \texttt{memory} in that order. Its
prompt combines the previous record, current neighbor messages, and the agent's
opening text $o_i^{(0)}$, which carries its intrinsic stance in words (the number $s_i$
is never shown to the model):
\[
P_i^{(t)}
=\operatorname{Prompt}\!\left(
\mathbf{x}_i^{(t)},\{\mathbf{x}_j^{(t)}\}_{j\in S_i^{(t)}},o_i^{(0)}
\right).
\]
An autoregressive LLM with parameters $\theta$ generates
$\mathbf{x}_i^{(t+1)}=(y_1,\ldots,y_m)$ according to
\[
\mathbb{P}_{\theta}\!\left(\mathbf{x}_i^{(t+1)}\mid P_i^{(t)}\right)
=\prod_{k=1}^{m}
\mathbb{P}_{\theta}\!\left(y_k\mid y_{<k},P_i^{(t)}\right).
\]
Given a temperature $\tau>0$, each next-token probability is the softmax
\[
\mathbb{P}_{\theta}\!\left(y_k\mid y_{<k},P_i^{(t)}\right)
=\frac{\exp\!\left(g_{\theta}\!\left(y_k\mid y_{<k},P_i^{(t)}\right)/\tau\right)}
{\sum_{y\in\mathcal{V}}\exp\!\left(g_{\theta}\!\left(y\mid y_{<k},P_i^{(t)}\right)/\tau\right)},
\]
where $g_{\theta}(\cdot)$ denotes the unnormalized logit output by the
LLM. This temperature-$\tau$ softmax sampling induces a stochastic behavior
policy over $\mathcal{Z}$,
\[
\pi_i^L\!\left(z\mid S^{(t)}\right)
=\sum_{\mathbf{x}\in\mathcal{V}^{\le M}:\;e(o(\mathbf{x}))=z}
\mathbb{P}_{\theta}\!\left(\mathbf{x}\mid P_i^{(t)}\right),
\]
where $o(\mathbf{x})$ is the \texttt{opinion} field of $\mathbf{x}$. The
simulations use deterministic decoding (the limit $\tau\to0$). Each Type-L agent makes one LLM call
per step, and each call rewrites the
\texttt{memory} field. Only the free-form
\texttt{opinion} field $o_i^{(t+1)}$ is passed to the semantic mapping, giving
$z_i^{(t+1)}=e(o_i^{(t+1)})\in[-1,+1]$.

\subsection{Coevolutionary Rewiring and Transition Kernel}
\label{sec:KNN}

After all opinions are updated, the network changes from $G^{(t)}$ to
$G^{(t+1)}$. Let
$d_{ij}^{(t+1)}=|z_j^{(t+1)}-s_i|$ be the distance between peer $j$'s
expressed opinion and agent $i$'s intrinsic opinion. A general stochastic
rewiring kernel assigns attachment weight
\[
w_{ij}^{(t+1)}=\exp\!\left(-\beta d_{ij}^{(t+1)}\right),
\qquad \beta>0,
\]
and draws each $K$-subset $S$ with probability
\[
\mathbb{P}\!\left(S_i^{(t+1)}=S\mid\mathbf{z}^{(t+1)},s_i\right)
=
\frac{\prod_{j\in S}w_{ij}^{(t+1)}}
{\displaystyle\sum_{\substack{S'\subset\mathcal{N}\setminus\{i\}\\|S'|=K}}
 \prod_{k\in S'}w_{ik}^{(t+1)}}.
\]
The experiments use the $\beta\rightarrow\infty$ limit:
\begin{equation}
\label{eq:knn}
S_i^{(t+1)}
=\operatorname*{arg\,top\text{-}K}_{j\ne i}
\left(-|s_i-z_j^{(t+1)}|\right),
\qquad K=5.
\end{equation}
Each agent therefore connects to the five peers whose expressed opinions lie nearest its own
intrinsic stance, with ties broken uniformly at random (seeded). Lemma~\ref{lem:knn}
characterizes this limit. Rewiring is each agent's own choice of neighbors, made by this fixed rule; in the Markov-game formulation of Section~\ref{sec:markov-mapping} it enters as the state transition.

Under conditional independence across agents, the transition kernel of the full state
combines the rewiring draws with the Type-L texts:
\[
\mathcal{P}\!\left(S^{(t+1)}\mid S^{(t)},\mathbf{a}^{(t)}\right)
=\left(\prod_{i=1}^{N}
\mathbb{P}\!\left(S_i^{(t+1)}\mid\mathbf{z}^{(t+1)},s_i\right)\right)
\times\mathbf{1}\!\left[\mathbf{X}^{(t+1)}
=\left(\mathbf{a}_i^{(t)}\right)_{i\in\mathcal{N}_L}\right]
\times\mathbf{1}\!\left[\mathbf{z}^{(t+1)}=\mathbf{z}(\mathbf{a}^{(t)})\right],
\]
where $\mathbf{z}(\mathbf{a}^{(t)})$ is the opinion profile the joint action is scored to.
Section~\ref{sec:markov-mapping} shows that this state space is finite.

\subsection{Social Cost and Price of Anarchy}
\label{sec:framework-cost}

Regardless of agent type, agent $i$ incurs the stage cost
\[
C_i\!\left(G^{(t)},\mathbf{z}^{(t)}\right)
=
\sum_{j\in S_i^{(t)}}(z_i^{(t)}-z_j^{(t)})^2
+\rho_iK(z_i^{(t)}-s_i)^2.
\]
The global social cost is the sum of individual costs:
\begin{equation}
\label{eq:sc}
SC(G,\mathbf{z})
=
\underbrace{\sum_i\sum_{j\in S_i}(z_i-z_j)^2}_{\text{disagreement}}
+
\underbrace{\sum_i\rho_iK(z_i-s_i)^2}_{\text{conformity}}.
\end{equation}
The first component measures disagreement with current neighbors, while the
second measures displacement from intrinsic opinions.

\begin{definition}[Nash equilibrium]
A joint strategy profile
$\boldsymbol{\pi}^*=(\pi_1^*,\ldots,\pi_N^*)$ is a Nash equilibrium if:
\begin{enumerate}
    \item for every Type-C agent $i\in\mathcal{N}_C$, $\pi_i^*$ minimizes
    its expected discounted cumulative stage cost given
    $\boldsymbol{\pi}_{-i}^*$ from every initial state $S^{(0)}$,
    \[
    \pi_i^*\in\arg\min_{\pi_i}
    \mathbb{E}_{(\pi_i,\boldsymbol{\pi}_{-i}^*)}\!\left[
    \sum_{t=0}^{\infty}\gamma^t
    C_i(G^{(t)},\mathbf{z}^{(t+1)})\mid S^{(0)}
    \right];
    \]
    \item for every Type-L agent $j\in\mathcal{N}_L$, $\pi_j^*$ is induced
    by the autoregressive LLM kernel
    $\mathbb{P}_{\theta}(\mathbf{x}\mid P_j)$.
\end{enumerate}
\end{definition}

The stage cost is charged on the graph the action faces. For $\gamma=0$ the Type-C condition
is the stage best response of Lemma~\ref{lem:br} on $G^{(t)}$; Eq.~\eqref{eq:type-c-update}
iterates this best response against the previous opinions, and the two coincide on a
stationary graph (Lemma~\ref{lem:fixed}). Under assumption~(i) of
Section~\ref{sec:assumptions} the LLM kernel is a point of the Type-L policy simplex, so the
Nash gap of Theorem~\ref{thm:equilibrium} is taken over all players.

\begin{definition}[Generalized Price of Anarchy]
\label{sec:poa}
Let $\mathcal{E}_{\mathrm{NE}}(\alpha)$ denote the stationary state-action
distributions resulting from an equilibrium at mixture level $\alpha$. The
generalized hybrid PoA is
\begin{equation}
\label{eq:PoA}
\mathrm{PoA}(\alpha)
=
\frac{\displaystyle
\max_{\nu\in\mathcal{E}_{\mathrm{NE}}(\alpha)}
\mathbb{E}_{(G,\mathbf{z})\sim\nu}[SC(G,\mathbf{z})]}
{\displaystyle\min_{G^*,\mathbf{z}^*}SC(G^*,\mathbf{z}^*)}.
\end{equation}
The denominator ranges over all $K$-out graphs, a relaxation of the $K$-NN optimum of
\citet{bhawalkar2013}.
\end{definition}

This generalized definition allows $\mathrm{PoA}(\alpha)$ to be analyzed and computed over
the whole range $\alpha\in[0,1]$, and it separates the effect of generative LLM reasoning
on social welfare and network stability.

The experiments report a topology-conditioned operational measure. At step
$t$, the planner optimizes opinions on the same realized graph:
\[
\mathbf{z}_t^{\mathrm{opt}}
=\arg\min_{\mathbf{z}\in[-1,+1]^N}SC(G^{(t)},\mathbf{z}),
\]
and
\begin{equation}
\label{eq:poa-operational}
\mathrm{PoA}^{(t)}
=
\frac{SC(G^{(t)},\mathbf{z}^{(t)})}
     {SC(G^{(t)},\mathbf{z}_t^{\mathrm{opt}})}
\geq1.
\end{equation}
Here $G^{(t)}$ is the graph rebuilt from $\mathbf{z}^{(t)}$, so $\mathrm{PoA}^{(t)}$ is
evaluated after rewiring. Lemma~\ref{lem:opt} gives the minimizer in closed form and shows
that this ratio is well defined. Since $SC(G,\mathbf{z}^{\mathrm{opt}}(G))\ge\min_{G^*,\mathbf{z}^*}SC(G^*,\mathbf{z}^*)$ for
every $G$, every $\nu\in\mathcal{E}_{\mathrm{NE}}(\alpha)$ satisfies
$\mathbb{E}_{\nu}\big[SC(G,\mathbf{z})/SC(G,\mathbf{z}^{\mathrm{opt}}(G))\big]\le\mathrm{PoA}(\alpha)$,
so the expected topology-conditioned ratio is a lower bound on Eq.~\eqref{eq:PoA}. Fixing the graph isolates the welfare loss associated with decentralized
opinion updating from that associated with the topology produced by rewiring.

\section{Equilibrium Analysis}
\label{sec:equilibrium}
\begingroup
\setlength{\emergencystretch}{1.5em}
In stochastic coevolutionary opinion formation games, network topologies evolve through
randomized rewiring, so small shifts in expressed opinions may drastically alter global
transition probabilities~\citep{chen:lu:lin:shi,chen:lu:lin:shi:hung}. By framing the
hybrid population of analytical cost-minimizing agents ($\mathcal{N}_C$) and LLM generative
agents ($\mathcal{N}_L$) as a multi-player general-sum Markov game and combining this
formulation with Theorem~2 of \citet{chen:lu:lin:shi:hung}, we show that Optimistic Gradient
Ascent (OGA) learning dynamics reach an approximate Nash equilibrium or, otherwise, a bounded
price of anarchy.

\subsection{Properties of the Opinion Stage}
\label{sec:stage}

Throughout this subsection we fix a step and suppress the superscript $(t)$ where no
confusion arises. We write $z^{\mathrm{BR}}$ for a best response, $\mathbf{z}^{\mathrm{NE}}$
for a Nash fixed point and $\mathbf{z}^{\mathrm{opt}}$ for the planner's optimum.

\begin{lemma}[Type-C best response]
\label{lem:br}
Fix a Type-C agent $i$ with out-neighborhood $S_i$, $|S_i|=K$, and the opinions
$\mathbf{z}_{-i}$ of the other agents. The function
\[
c_i(z)=\sum_{j\in S_i}(z-z_j)^2+\rho_iK(z-s_i)^2,
\]
which is the stage cost of Section~\ref{sec:framework-cost} with $z_i$ replaced by the free
variable $z$, is strictly convex and has the unique minimizer
\[
z_i^{\mathrm{BR}}=\frac{\sum_{j\in S_i}z_j+\rho_iKs_i}{K(1+\rho_i)} .
\]
\end{lemma}

\begin{proof}
The function $c_i$ is a sum of $K+1$ squares, so $c_i''(z)=2K+2\rho_iK=2K(1+\rho_i)$, which is
positive because $\rho_i>0$. Hence $c_i$ is strictly convex. Since $c_i(z)\to\infty$ as
$|z|\to\infty$, a minimizer exists, and strict convexity makes it unique. Setting
$c_i'(z)=2\sum_{j\in S_i}(z-z_j)+2\rho_iK(z-s_i)$ to zero gives
$K(1+\rho_i)\,z=\sum_{j\in S_i}z_j+\rho_iKs_i$, which is the stated expression.
\end{proof}

\begin{lemma}[Invariance of the opinion interval]
\label{lem:inv}
If $s_i\in[-1,1]$ for all $i$, then $z_i^{(t)}\in[-1,1]$ for all $i\in\mathcal{N}$ and all
$t\ge0$.
\end{lemma}

\begin{proof}
For a Type-L agent, $z_i^{(0)}=s_i\in[-1,1]$, and for every opinion text $o$ the score
$e(o)=\tanh\!\big(\mathbf{W}\,\mathbf{h}_{[\mathrm{CLS}]}(o)\big)$ of Section~\ref{sec:stance}
lies in $(-1,1)$, independently of $t$. For a Type-C agent we argue by induction on $t$. At
$t=0$, $z_i^{(0)}=s_i\in[-1,1]$. Suppose $z_j^{(t)}\in[-1,1]$ for all $j$. By
Lemma~\ref{lem:br}, $z_i^{(t+1)}$ is a weighted combination of
$\{z_j^{(t)}\}_{j\in S_i^{(t)}}$ and $s_i$ with weight $\tfrac{1}{|S_i^{(t)}|(1+\rho_i)}$ on each
of the $|S_i^{(t)}|$ neighbors and $\tfrac{\rho_i}{1+\rho_i}$ on the intrinsic opinion (and
$z_i^{(t+1)}=s_i$ if $S_i^{(t)}=\emptyset$). These weights are nonnegative and sum to
$\tfrac{1}{1+\rho_i}+\tfrac{\rho_i}{1+\rho_i}=1$, so
$z_i^{(t+1)}$ is a convex combination of points of $[-1,1]$ and lies in $[-1,1]$.
\end{proof}

Consequently the action space of a Type-C agent may be taken as
$\mathcal{Z}\subset[-1,1]$ without loss of generality.

\begin{lemma}[Equilibrium of the opinion stage on a fixed graph]
\label{lem:fixed}
Fix $G$ and the Type-L opinions $\mathbf{z}_L\in[-1,1]^{|\mathcal{N}_L|}$. Define
$\Phi:[-1,1]^{|\mathcal{N}_C|}\to[-1,1]^{|\mathcal{N}_C|}$ by
\[
\Phi_i(\mathbf{z}_C)=\frac{\sum_{j\in S_i}z_j+\rho_iKs_i}{K(1+\rho_i)},
\qquad i\in\mathcal{N}_C,
\]
where $z_j$ is taken from $\mathbf{z}_L$ for $j\in\mathcal{N}_L$ (we write
$\Phi(\mathbf{z}_C;\mathbf{z}_L)$ when the dependence on $\mathbf{z}_L$ matters). Then
\begin{enumerate}
\item $\Phi$ is a contraction in $\|\cdot\|_\infty$ with modulus
      $q=\max_{i\in\mathcal{N}_C}\frac{1}{1+\rho_i}<1$;
\item $\Phi$ has a unique fixed point $\mathbf{z}_C^{\mathrm{NE}}$, which is the unique Nash
      equilibrium of the opinion stage among Type-C agents given $(G,\mathbf{z}_L)$;
\item the synchronous iterates $\mathbf{z}_C^{(t+1)}=\Phi(\mathbf{z}_C^{(t)})$ satisfy
      $\|\mathbf{z}_C^{(t)}-\mathbf{z}_C^{\mathrm{NE}}\|_\infty\le 2q^{t}$, with $t$ counted
      from the first step on $G$;
\item if the Type-L opinions vary, $\mathbf{z}_C^{(t+1)}=\Phi(\mathbf{z}_C^{(t)};\mathbf{z}_L^{(t)})$,
      and $\mathbf{z}_C^{\mathrm{NE}}$ is taken at $\mathbf{z}_L$, then
      $\|\mathbf{z}_C^{(t)}-\mathbf{z}_C^{\mathrm{NE}}\|_\infty\le
      \max\{2q^{t},\,q\max_{u<t}\|\mathbf{z}_L^{(u)}-\mathbf{z}_L\|_\infty\}$.
\end{enumerate}
For $\alpha=1$ ($\mathcal{N}_L=\emptyset$), $\mathbf{z}^{\mathrm{NE}}$ is the
Friedkin--Johnsen equilibrium on $G$.
\end{lemma}

\begin{proof}
$\Phi$ maps $[-1,1]^{|\mathcal{N}_C|}$ into itself by the convex-combination argument of
Lemma~\ref{lem:inv}. For $\mathbf{z}_C,\mathbf{z}'_C$ and $i\in\mathcal{N}_C$, the intrinsic
term and the Type-L neighbors cancel, so only Type-C neighbors contribute:
\[
\begin{aligned}
|\Phi_i(\mathbf{z}_C)-\Phi_i(\mathbf{z}'_C)|
&\le\frac{1}{K(1+\rho_i)}\sum_{j\in S_i\cap\mathcal{N}_C}|z_j-z'_j| \\
&\le\frac{|S_i\cap\mathcal{N}_C|}{K(1+\rho_i)}\,\|\mathbf{z}_C-\mathbf{z}'_C\|_\infty
\le\frac{1}{1+\rho_i}\,\|\mathbf{z}_C-\mathbf{z}'_C\|_\infty ,
\end{aligned}
\]
which proves (1). The same estimate holds on all of $\mathbb{R}^{|\mathcal{N}_C|}$. Since
$[-1,1]^{|\mathcal{N}_C|}$ is complete, Banach's fixed-point theorem gives a unique fixed
point, which by the same estimate is also unique in $\mathbb{R}^{|\mathcal{N}_C|}$, and
$\|\mathbf{z}_C^{(t)}-\mathbf{z}_C^{\mathrm{NE}}\|_\infty
\le q^{t}\|\mathbf{z}_C^{(0)}-\mathbf{z}_C^{\mathrm{NE}}\|_\infty\le 2q^{t}$, which proves
(3). A profile is a fixed point of $\Phi$ if and only if every Type-C agent plays its best
response to the others (Lemma~\ref{lem:br}), that is, if and only if it is a Nash equilibrium
of the opinion stage; this proves (2). For (4), the estimate of (1) extends to
$|\Phi_i(\mathbf{z}_C;\mathbf{z}_L)-\Phi_i(\mathbf{z}'_C;\mathbf{z}'_L)|\le\frac{1}{1+\rho_i}
\max\{\|\mathbf{z}_C-\mathbf{z}'_C\|_\infty,\|\mathbf{z}_L-\mathbf{z}'_L\|_\infty\}$, so the error
$e_t=\|\mathbf{z}_C^{(t)}-\mathbf{z}_C^{\mathrm{NE}}\|_\infty$ satisfies
$e_{t+1}\le q\max\{e_t,\|\mathbf{z}_L^{(t)}-\mathbf{z}_L\|_\infty\}$, and induction from
$e_0\le2$ gives $e_t\le\max\{2q^t,\,q\max_{u<t}\|\mathbf{z}_L^{(u)}-\mathbf{z}_L\|_\infty\}$.
\end{proof}

\begin{lemma}[$K$-nearest-neighbor rewiring as the high-intensity limit]
\label{lem:knn}
For agent $i$ and a $K$-subset $S'\subset\mathcal{N}\setminus\{i\}$, let
$D_i(S')=\sum_{j\in S'}d_{ij}$ with $d_{ij}=|z_j-s_i|$, let $D_i^{*}=\min_{S'}D_i(S')$ and
$\mathcal{M}_i=\arg\min_{S'}D_i(S')$, and order the distances from $i$ as
$d_{(1)}\le d_{(2)}\le\cdots\le d_{(N-1)}$. As $\beta\to\infty$, the rewiring probability of
Section~\ref{sec:KNN} tends to $1/|\mathcal{M}_i|$ for $S'\in\mathcal{M}_i$ and to $0$
otherwise. A $K$-subset belongs to $\mathcal{M}_i$ if and only if it contains every peer with
$d_{ij}<d_{(K)}$ and fills its remaining places with peers at distance exactly $d_{(K)}$. If
$d_{(K)}<d_{(K+1)}$, $\mathcal{M}_i$ is the single $K$-nearest-neighbor set; otherwise the
limit fills the remaining places uniformly at random from the tied peers.
\end{lemma}

\begin{proof}
Multiplying the numerator and the denominator of the rewiring probability by
$\exp(\beta D_i^{*})$ gives
\[
P_\beta(S')=\frac{\exp\!\big(-\beta\,(D_i(S')-D_i^{*})\big)}
                 {\sum_{|S''|=K}\exp\!\big(-\beta\,(D_i(S'')-D_i^{*})\big)} .
\]
Every exponent is now nonpositive. A subset in $\mathcal{M}_i$ contributes exactly $1$, and a
subset outside it has $D_i(S'')-D_i^{*}>0$ and contributes a term that falls to $0$ as
$\beta\to\infty$. The denominator therefore tends to $|\mathcal{M}_i|$, while the numerator
tends to $1$ when $S'\in\mathcal{M}_i$ and to $0$ otherwise.

It remains to identify $\mathcal{M}_i$. Since $D_i(S')$ sums $K$ of the ordered distances,
$D_i(S')\ge\sum_{r=1}^{K}d_{(r)}$ for every $K$-subset, and a subset of the stated form attains
that bound, so $D_i^{*}=\sum_{r=1}^{K}d_{(r)}$ and every such subset lies in $\mathcal{M}_i$.
Conversely, let $S'$ attain $D_i^{*}$. If $S'$ contains a peer $u$ with $d_{iu}>d_{(K)}$, at
most $K-1$ of its members lie within $d_{(K)}$, so some peer $v\notin S'$ has
$d_{iv}\le d_{(K)}<d_{iu}$, and exchanging $u$ for $v$ lowers $D_i$. If $S'$ omits a peer $v$
with $d_{iv}<d_{(K)}$, then, because at most $K-1$ peers lie strictly closer than $d_{(K)}$,
$S'$ contains a peer $u$ with $d_{iu}=d_{(K)}$, and exchanging $u$ for $v$ again lowers
$D_i$. Both cases contradict minimality, so every minimizer has the stated form. When
$d_{(K)}<d_{(K+1)}$, exactly $K$ peers lie within $d_{(K)}$ and $|\mathcal{M}_i|=1$;
otherwise the limit spreads its mass evenly over the subsets that differ only in which of the
tied peers they take.
\end{proof}

The kernel factorizes across agents (Eq.~\eqref{eq:kernel}), so the joint limit is the
product of the per-agent limits.

For a fixed $G$, write $m_{ij}=\mathbf{1}[j\in S_i]+\mathbf{1}[i\in S_j]$ for the weights of
the undirected multigraph induced by $G$, $L$ for its Laplacian, and
$R=\mathrm{diag}(\rho_1,\dots,\rho_N)$. Each unordered pair $\{i,j\}$ appears $m_{ij}$ times
in the disagreement term of Eq.~\eqref{eq:sc}, so
\[
SC(G,\mathbf{z})=\mathbf{z}^{\top}L\mathbf{z}+(\mathbf{z}-\mathbf{s})^{\top}KR\,(\mathbf{z}-\mathbf{s}).
\]

\begin{lemma}[The planner optimum and well-posedness of $\mathrm{PoA}^{(t)}$]
\label{lem:opt}
For every $G$:
\begin{enumerate}
\item $SC(G,\cdot)$ is strictly convex and has the unique minimizer
      $\mathbf{z}^{\mathrm{opt}}=(L+KR)^{-1}KR\,\mathbf{s}$;
\item each $z_i^{\mathrm{opt}}$ is a convex combination of $s_1,\dots,s_N$; in particular
      $\mathbf{z}^{\mathrm{opt}}\in[-1,1]^N$;
\item $SC^{\mathrm{opt}}=\min_{\mathbf{z}}SC(G,\mathbf{z})>0$ whenever some edge $j\in S_i$
      joins agents with $s_i\ne s_j$. In that case
      $\mathrm{PoA}^{(t)}=SC(G^{(t)},\mathbf{z}^{(t)})/SC^{\mathrm{opt}}$ is well defined and
      $\mathrm{PoA}^{(t)}\ge 1$.
\end{enumerate}
\end{lemma}

\begin{proof}
(1) The Hessian is $2(L+KR)$. Since $L\succeq 0$ and, because every $\rho_i>0$, $KR\succ 0$,
we have $L+KR\succ 0$, so $SC$ is strictly convex. Setting the gradient $2L\mathbf{z}+2KR(\mathbf{z}-\mathbf{s})$
to zero gives $(L+KR)\mathbf{z}=KR\mathbf{s}$.

(2) Let $M=L+KR$. Its off-diagonal entries $-m_{ij}$ are nonpositive and each row is strictly
diagonally dominant, since $M_{ii}=\sum_{j}m_{ij}+K\rho_i>\sum_{j\ne i}|M_{ij}|$. Hence $M$ is
a nonsingular M-matrix and $M^{-1}\ge 0$ entrywise, so $W=M^{-1}KR\ge 0$. Since
$L\mathbf{1}=0$, $M\mathbf{1}=KR\mathbf{1}$, and therefore $W\mathbf{1}=\mathbf{1}$. Each row
of $W$ is thus a probability vector, and $\mathbf{z}^{\mathrm{opt}}=W\mathbf{s}$ is a row-wise
convex combination of $\mathbf{s}$.

(3) Both terms of $SC$ are nonnegative, and $SC(G,\mathbf{z})=0$ requires $z_i=s_i$ for all
$i$ and $z_i=z_j$ for every edge, hence $s_i=s_j$ for every edge. If some edge has
$s_i\ne s_j$, then $SC(G,\mathbf{z})>0$ for all $\mathbf{z}$, and since the minimum is attained
by (1), $SC^{\mathrm{opt}}>0$. Finally
$SC(G^{(t)},\mathbf{z}^{(t)})\ge SC^{\mathrm{opt}}$ by definition of the minimum.
\end{proof}

In the sampled populations no three agents share an intrinsic opinion and $K=5$, so every
realized graph contains an edge joining different intrinsic opinions and the stepwise PoA is
well defined.

\begin{remark}[Cost decomposition]
\label{rem:decomp}
Writing Eq.~\eqref{eq:sc} as $SC=\mathcal{D}+\mathcal{C}$ (disagreement and conformity) and
dividing by $SC^{\mathrm{opt}}$ gives
$\mathrm{PoA}^{(t)}=\mathcal{D}^{(t)}/SC^{\mathrm{opt}}+\mathcal{C}^{(t)}/SC^{\mathrm{opt}}$.
\end{remark}

\begin{remark}[The $9/8$ bound]
\label{rem:98}
At $\alpha=1$ H-COG is the $K$-NN game of \citet{bhawalkar2013} with heterogeneous
stubbornness; on a stationary graph Lemmas~\ref{lem:fixed} and~\ref{lem:opt} identify
$\mathrm{PoA}^{(t)}$ with the equilibrium-to-optimum ratio of the Friedkin--Johnsen opinion
game on $G^{(t)}$.
\citet{bindel2015bad} and \citet{bhawalkar2013} bound this ratio by $9/8$ on undirected graphs,
and about 85\% of the edges in H-COG graphs are reciprocal.
\end{remark}

\begin{lemma}[Polarization and conformity cost]
\label{lem:floor}
Let $\sigma_s$ and $\sigma_z$ be the population standard deviations of intrinsic and expressed
opinions, so that $\sigma_z^2=P_z$ (Section~\ref{sec:metrics}), and let $\rho_{\min}$ be the
smallest stubbornness coefficient in the population. For every graph and every opinion
profile,
\[
\sum_i\rho_iK(z_i-s_i)^2\;\ge\;\rho_{\min}KN(\sigma_s-\sigma_z)^2 .
\]
\end{lemma}
\begin{proof}
The conformity term is at least $\rho_{\min}K\|\mathbf{z}-\mathbf{s}\|^2$. Projecting
$\mathbf{z}-\mathbf{s}$ onto the subspace of zero-mean vectors does not increase its length,
so $\|\mathbf{z}-\mathbf{s}\|\ge\|(\mathbf{z}-\bar z\mathbf{1})-(\mathbf{s}-\bar s\mathbf{1})\|$,
and by the reverse triangle inequality the right-hand side is at least
$\bigl|\|\mathbf{z}-\bar z\mathbf{1}\|-\|\mathbf{s}-\bar s\mathbf{1}\|\bigr|=\sqrt{N}\,|\sigma_z-\sigma_s|$.
\end{proof}
The bound does not depend on how opinions are formed: it holds for Type-C, Type-L and mixed
populations. The more concentrated the expressed opinions are relative to the intrinsic ones,
the higher the floor on the conformity cost.

\subsection{The Hybrid Game as a Markov Game}
\label{sec:markov-mapping}

We write the Hybrid Coevolutionary Opinion Game as a multi-player general-sum
Markov game
\[
\mathcal{M}_{\text{hybrid}}
=\left\langle \mathcal{N},\mathcal{S},\mathcal{A},\mathcal{P},
\{r^i\}_{i\in\mathcal{N}},\gamma\right\rangle ,
\]
whose components are the following.

\paragraph{Players.} $\mathcal{N}=\mathcal{N}_C\cup\mathcal{N}_L$, with
$|\mathcal{N}_L|=\operatorname{round}\big((1-\alpha)N\big)$ Type-L agents and
$|\mathcal{N}_C|=N-|\mathcal{N}_L|$ Type-C agents.

\paragraph{States.} A global state at step $t$ is the triple
$S^{(t)}=(G^{(t)},\mathbf{z}^{(t)},\mathbf{X}^{(t)})$ of Section~\ref{sec:framework-updates}:
the active out-edges, the opinions from which they were built, and the Type-L texts that
produced those opinions. Generated texts have bounded length and opinions are represented at
finite precision, so $\mathcal{S}$ is finite, the property the cited result requires.

\paragraph{Actions.} $\mathcal{A}=\prod_{i=1}^{N}\mathcal{A}_i$. A Type-C agent
plays $a_i=z_i\in\mathcal{Z}$ directly; a Type-L agent plays
$a_i=\mathbf{x}_i\in\mathcal{V}^{\le M}$, which the semantic mapping carries to
$z_i=e(o(\mathbf{x}_i))\in[-1,+1]$.

\paragraph{Transition kernel.} The opinion and text components of the next state are
determined by the joint action, so the kernel of Section~\ref{sec:KNN} is governed by its
graph factor. Write $\mathbf{z}^{(t+1)}$ for the
opinions the joint action $\mathbf{a}^{(t)}$ is scored to, and $\mathbf{s}$ for the fixed
intrinsic opinions. Rewiring snapshots those opinions and then lets every agent rank the
same snapshot on its own, so the $N$ draws are conditionally independent and
\begin{equation}
\label{eq:kernel}
\mathcal{P}\!\left(G^{(t+1)}\mid S^{(t)},\mathbf{a}^{(t)}\right)
=\prod_{i=1}^{N}
\mathbb{P}\!\left(S_i^{(t+1)}\mid \mathbf{z}^{(t+1)},s_i\right),
\end{equation}
with each factor the Gibbs form of Section~\ref{sec:KNN}. Two properties of
Eq.~\eqref{eq:kernel} matter for what follows. It does not depend on
$G^{(t)}$, because the rule rebuilds every out-neighborhood from the current
opinions rather than editing the previous edge set, so the chain is driven by
the opinions alone. And in the $\beta\rightarrow\infty$ limit used throughout
the experiments each factor collapses to the uniform distribution over the minimizers
$\mathcal{M}_i$ of the total distance (Lemma~\ref{lem:knn}).

\paragraph{Rewards.} The reward is the negative stage cost of the opinions the joint action
is scored to, on the graph that action faces,
$r^i(S^{(t)},\mathbf{a}^{(t)})=-C_i(G^{(t)},\mathbf{z}^{(t+1)})$, and since
$z_i,s_i\in[-1,+1]$ it is bounded by $|r^i|\le 4(1+\rho_i)K$.

\subsection{Structural Assumptions}
\label{sec:assumptions}

The learning dynamics of the hybrid Markov coevolutionary opinion game incorporate the
structural assumptions of \citet{chen:lu:lin:shi:hung}. We write $\sigma\in\mathcal{S}$ for
states and $k$ for OGA iterations.

\begin{enumerate}[label=(\roman*)]
    \item \textbf{Norm Pairs and Diameter Bounds:}
    Generated texts have bounded length and opinions are represented at finite precision,
    so each action set $\mathcal{A}_i$ is finite and each state-conditioned policy space
    $\mathcal{X}_i^\sigma$ is the probability simplex over $\mathcal{A}_i$. Under the
    Euclidean norm, which is its own dual, the diameter of the simplex gives
    \[
    \Omega_i^\sigma \triangleq \sup_{x, x' \in \mathcal{X}_i^\sigma} \|x - x'\| = \sqrt{2} < \infty ,
    \]
    and the dual norm pair satisfies $\|\cdot\| \ge C \|\cdot\|_1$ and
    $\|\cdot\|_* \le C^* \|\cdot\|_\infty$ with $C=\min_i|\mathcal{A}_i|^{-1/2}$ and
    $C^*=\max_i|\mathcal{A}_i|^{1/2}$~\citep{chen:lu:lin:shi}.

    \item \textbf{Smooth Regularizer and OGA Updates:}
    Each player updates its state-conditioned strategy $\pi_{i,k}^\sigma$ using
    Optimistic Gradient Ascent with $\ell_i$-smooth regularizers
    $\mathcal{R}_i$~\citep{chen:lu:lin:shi}:
    \begin{align*}
        \hat{x}_{i,k+1}^\sigma &= \Pi_{\mathcal{X}_i^\sigma}\left\{ \hat{x}_{i,k}^\sigma + \eta u_{i,k}^\sigma \right\} \\
        x_{i,k+1}^\sigma &= \Pi_{\mathcal{X}_i^\sigma}\left\{ \hat{x}_{i,k+1}^\sigma + \eta u_{i,k}^\sigma \right\}
    \end{align*}
    where $\Pi_{\mathcal{X}_i^\sigma}$ is the Euclidean projection onto the simplex and
    $u_{i,k}^\sigma$ is the optimistic gradient estimator, with
    $\big\|u_{i,k}^\sigma-Q_{i,k}^\sigma\big(\textstyle\bigotimes_{j\ne i}\pi_{j,k}^\sigma\big)\big\|_\infty\le\epsilon$~\citep{chen:lu:lin:shi}. Here $Q_{i,k}^\sigma(\cdot)\in\mathbb{R}^{|\mathcal{A}_i|}$
    is player $i$'s state-action value vector at $\sigma$ against the other players' current
    policies, and the Nash gap of a profile at $\sigma$ is the largest gain that any single
    player can obtain at $\sigma$ by a unilateral change of its policy.
    With the Euclidean regularizer $\mathcal{R}_i(x)=\tfrac12\|x\|_2^2$, $\ell_i=1$.

    \item \textbf{Learning Rate Regime ($\eta$):}
    The learning rate satisfies $\eta \le O(1/\sqrt{T})$, which offsets the positive regret
    terms generated by time-varying $Q$-value functions in general-sum games
    ~\citep{chen:lu:lin:shi}.

    \item \textbf{Bounded Value Functions:}
    Under quadratic disagreement costs the rewards are bounded, so the state value functions
    $V_{\pi}^i(\sigma)$ are continuous and bounded~\citep{chen:lu:lin:shi}.
\end{enumerate}

The convergence result of \citet{chen:lu:lin:shi} (Theorem~1 there; Theorem~2 of
\citet{chen:lu:lin:shi:hung}) is a dichotomy on the sign of the cumulative regret
$\sum_iR_i^T$ of all players. If $\sum_iR_i^T\ge0$ for every $T$, some iterate is an
approximate Nash equilibrium. Otherwise the proof gives $\sum_iR_i^T\le-\epsilon^2T/(8\eta)$,
and the time-averaged social value of the learning dynamics is bounded below by a fraction of
the optimum, which bounds their price of anarchy. Theorem~\ref{thm:equilibrium} states both
cases for H-COG.

\subsection{Formal Equilibrium Statement}
\label{sec:thm}

\begin{theorem}[Equilibrium Guarantee in Hybrid Opinion Games]
\label{thm:equilibrium}
Consider the hybrid Markov coevolutionary opinion game played by Type-C and Type-L agents
running OGA dynamics with learning rate $\eta \le O(1/\sqrt{T})$ under assumptions (i)--(iv),
and let $\epsilon > 0$ and
\[
T \ge \frac{4}{\epsilon^2 |\mathcal{S}|} \sum_{\sigma \in \mathcal{S}} \sum_{i \in \mathcal{N}}
\big(\Omega_i^\sigma\big)^2=\frac{8N}{\epsilon^2}
\]
(the threshold required in the proof of Theorem~2 of \citet{chen:lu:lin:shi:hung}).
Then one of the following holds.
\begin{enumerate}[label=(\alph*)]
\item If $\sum_iR_i^T\ge0$ for every $T$, there exists an iteration $k \in [T]$ such that the
joint policy profile
$\boldsymbol{\pi}_k^{\sigma_k} = (\pi_{1,k}^{\sigma_k}, \dots, \pi_{N,k}^{\sigma_k})$ is an
approximate Nash equilibrium for the state $\sigma_k$: its Nash gap is at most a multiple of
$\epsilon$~\citep{chen:lu:lin:shi,chen:lu:lin:shi:hung},
\[
\mathrm{Nash\text{-}gap}(\boldsymbol{\pi}_k^{\sigma_k}) \le \epsilon \left( C^* \max_{i}
\Big\| Q_{i,k}^{\sigma_k}\big(\textstyle\bigotimes_{j\ne i}\pi_{j,k}^{\sigma_k}\big)\Big\|_\infty
+ \frac{2 \max_i \{\ell_i \Omega_i^{\sigma_k}\}}{\eta} \right).
\]
\item Otherwise, the time-averaged social value of the dynamics satisfies
\[
\frac{1}{T}\sum_{k=1}^{T} V^{\boldsymbol{\pi}_k}(\sigma) \;\ge\;
\frac{\lambda}{1+\mu}\,V^{\boldsymbol{\pi}^*}(\sigma)+\frac{\epsilon^2}{8\eta(1+\mu)},
\]
where $V^{\boldsymbol{\pi}}(\sigma)=\sum_i V^i_{\boldsymbol{\pi}}(\sigma)$ is the social value,
$\boldsymbol{\pi}^*$ maximizes it, and $\lambda,\mu$ are the constants of the bound of
\citet{chen:lu:lin:shi} (written $\alpha,\beta$ there; renamed here to avoid a clash with the
rewiring intensity $\beta$).
\end{enumerate}
\end{theorem}

The result extends to multiple players~\citep[Remark~1]{chen:lu:lin:shi:hung}.

In case~(a) Theorem~\ref{thm:equilibrium} guarantees an approximate equilibrium at some
iteration; in case~(b) the learning dynamics may keep moving, but their time-averaged social
value stays within a bounded factor of the optimum. The run classes of
Section~\ref{sec:stopping} describe what the coevolving graph does in the simulations. A
\texttt{fixed\_point} run stays
on a stationary graph, on which Lemma~\ref{lem:fixed}(3)--(4) drives the Type-C opinions
geometrically to the stage Nash equilibrium, up to $q$ times the residual variation of the
Type-L opinions. In \texttt{limit\_cycle} and \texttt{plateau} runs rewiring continues
without pointwise convergence, and the statistics of the dynamics are invariant, not the
state itself.

\subsection{Interpretation for Hybrid Networks}
\label{sec:interpretation}

Theorem~\ref{thm:equilibrium} has a direct consequence for hybrid opinion games: the
coevolutionary feedback loop between LLM prompt generation and continual $K$-NN rewiring
reaches an approximate Nash equilibrium at some iteration, at which no agent can lower its cost by more
than the Nash gap through a unilateral deviation, or, when the cumulative regret is negative,
keeps its time-averaged social value within a bounded factor of the optimum.
\endgroup

\section{Methodology}
\label{sec:method}
\subsection{Simulation Framework Overview}

\begingroup
\setlength{\emergencystretch}{1em}
Figure~\ref{fig:pipeline} summarizes the two-stage pipeline. \textbf{Stage~1} converts
Reddit comments into stance distributions and samples $N=50$ agents from them.
\textbf{Stage~2} assigns the agents to Type-C and Type-L according to $\alpha$ and runs
H-COG until the stopping rule of Section~\ref{sec:stopping} is met.
\par
\endgroup

\begin{figure}[t]
    \centering
    \includegraphics[width=0.92\linewidth]{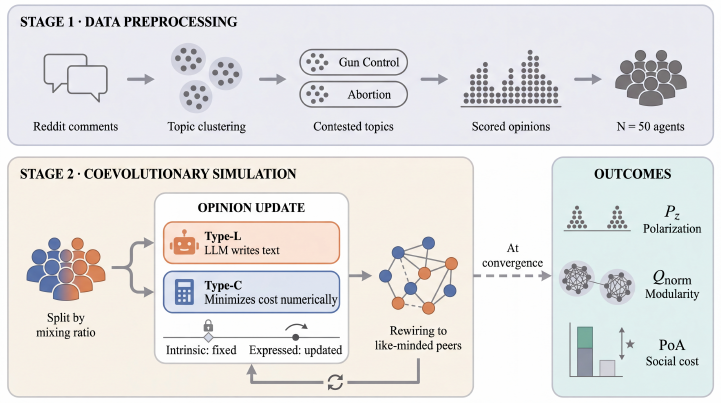}
    \caption{Two-stage pipeline: Stage~1 extracts stance distributions from
    Reddit; Stage~2 runs H-COG with heterogeneous agents on a coevolving
    network.}
    \label{fig:pipeline}
\end{figure}

\subsection{Stage 1: Reddit Data Pipeline}

\subsubsection*{Corpus Processing and Topic Clustering}
\label{sec:corpus}

We source top-level \texttt{r/politics} comments from the Pushshift dumps and keep only the
highest-scoring comment per author and topic. Comments are encoded with the Sentence-BERT model
\texttt{all-}\allowbreak\texttt{MiniLM-}\allowbreak\texttt{L6-v2}~\citep{reimers2019sentencebert}, reduced from 384 to 50 dimensions by
UMAP~\citep{mcinnes2018umap}, and clustered with
HDBSCAN~\citep{campello2013density,mcinnes2017hdbscan} (\texttt{min\_cluster\_size}$=50$,
\texttt{min\_samples}$=5$) into 1,184 topic clusters.

\subsubsection*{Continuous Stance Measurement}
\label{sec:stance}

\paragraph{Architecture.}
We replace the classification head of
\texttt{RobertaForSequenceClassification}~\citep{liu2019roberta} with a
regression head,
\begin{equation}
\hat{y}_i
=\tanh\!\left(\mathbf{W}\mathbf{h}_{[\mathrm{CLS}]}\right)
\in[-1,+1],
\end{equation}
where $\mathbf{h}_{[\mathrm{CLS}]}\in\mathbb{R}^{768}$ is the pooled token
representation and $\mathbf{W}\in\mathbb{R}^{1\times768}$. Negative and
positive endpoints represent strong opposition and support, respectively.

\paragraph{Training.}
The regressor is trained on the IBM Argument Quality 30K
corpus~\citep{gretz2020}, which contains crowd-sourced annotations for
contested topics including gun control and abortion. Graded supervision is
recovered through Best--Worst Scaling~\citep{kiritchenko2017bws}. On 120
held-out texts, the fitted scorer is monotone in group means and correlates
with the graded target at $r=0.82$. This regression head is the semantic mapping $e$ of
Section~\ref{sec:framework-agents}. The same instrument scores the Reddit
corpus, initializes intrinsic opinions, and measures Type-L expressed
opinions, placing all three on one common scale.

\subsubsection*{Stratified Agent Sampling}
\label{sec:stratified}

For each target topic, we select $N=50$ agents as follows:
\begin{enumerate}
    \item Score all comments to obtain the empirical stance distribution
    $\hat{P}(\mathrm{stance})$.
    \item Partition $[-1,+1]$ into five strata using thresholds
    $\{-0.75,-0.25,+0.25,+0.75\}$.
    \item For each stratum $b$, compute its empirical proportion
    $p_b$ and sample
    $n_b=\lfloor50p_b\rceil$ comments, adjusting rounding so that
    $\sum_b n_b=50$.
\end{enumerate}
Agent $i$ keeps the continuous score of its source comment as $s_i$.

\subsection{Stage 2: H-COG Simulation}

\subsubsection{Agent Initialization}
\label{ssec:agent-init}

Each agent receives an intrinsic opinion $s_i$ and a stubbornness
$\rho_i\sim\mathrm{Uniform}(0.3,0.7)$ drawn independently of stance. A Type-C agent
starts at $z_i^{(0)}=s_i$. A Type-L agent also carries a fixed persona written by
Phi-4~\citep{abdin2024phi4}, giving it a name, an age, a gender, an education level and five
personality traits, following the persona conditioning of~\citet{park2023}. The same model writes its opening text from that persona
and the stance keyword of its stratum. Like a Type-C agent, a Type-L agent starts at
$z_i^{(0)}=s_i$; from the first step on, its expressed opinion is the score
$e(o_i^{(t)})$ of the text it writes.

\subsubsection{Network Initialization}
\label{sec:config-net}

We manipulate the initial topology using three graph families:
Barab\'asi--Albert with preferential attachment $m=2$
\citep{barabasi1999emergence}, Erd\H{o}s--R\'enyi~\citep{erdos1959random}, and
Watts--Strogatz~\citep{watts1998collective}. The three families are matched on
mean degree while preserving differences in clustering (Section~\ref{sec:config}).
Ten independently seeded graphs from
each family are crossed with nine $\alpha$ levels and two topics, giving
$9\times3\times10\times2=540$ runs.

$K$-NN rewiring (Eq.~\eqref{eq:knn}) rebuilds all edges after every update, so the initial topology
affects the first update directly and later dynamics only through that
trajectory. Across seeds, we vary the topology and Type-C/Type-L assignment
while holding the sampled agents, intrinsic opinions, personas, and
stubbornness values fixed within each topic, so confidence intervals reflect topology and
assignment variability, not resampling of users.

\subsubsection{Opinion Updates}

Type-C agents use the best response of Eq.~\eqref{eq:type-c-update}. For Type-L agents, whose record is
(\texttt{reasoning}, \texttt{opinion}, \texttt{memory}) (Section~\ref{sec:framework-updates}),
placing \texttt{reasoning} first conditions the opinion on the preceding
deliberation~\citep{wei2022chain}, and a fixed field order keeps the response format stable
across steps~\citep{turpin2023language}. The \texttt{opinion} field is the agent's only
position output and is scored by $e$. We do not request a numerical self-rating because
language models regress toward the center of an offered scale~\citep{llm2025likert} and
anchor on round values~\citep{llm2026numericbias}. Rewriting \texttt{memory} at each step
bounds prompt growth, and current neighbor opinions are supplied verbatim.

\subsection{Structural Convergence and Stopping Rule}

\subsubsection{Convergence Criterion}
\label{sec:conv-motivation}

We define convergence operationally as entry into an attractor or a stable
structural regime: the rate of graph change ceases to decrease, although
individual edges need not stop changing. $K$-NN reconstructs each
out-neighborhood from a continuous ranking, so a small perturbation near the
cutoff can replace an edge. Adaptive networks may consequently settle into
dynamic regimes rather than frozen configurations~\citep{gross2008adaptive}.
Our criterion detects either an exactly recurring edge set or a plateau in
the rate of structural change.

\subsubsection{Structural Convergence Metrics}
\label{sec:conv-metrics}

Let $E(t)$ be the directed edge set and $S_i^{(t)}$ agent $i$'s
out-neighborhood.

\paragraph{Temporal correlation coefficient.}
Following~\citet{tang2010small} and its directed
extension~\citep{buttner2016temporal},
\[
C_i^{\mathrm{out}}(t)
=\frac{|S_i^{(t-1)}\cap S_i^{(t)}|}
{\sqrt{d_i^{\mathrm{out}}(t-1)d_i^{\mathrm{out}}(t)}},
\qquad
C^{\mathrm{out}}(t)
=\frac{1}{N}\sum_i C_i^{\mathrm{out}}(t).
\]
Because every out-degree equals $K$ from $t=1$ on, for $t\ge2$ this is the mean fraction of neighbors
retained between consecutive steps.

\paragraph{Secondary structural measures.}
We also log the normalized Hamming distance
$dS(t)=|E(t)\triangle E(t-1)|/(|E(t)|+|E(t-1)|)$, the $\ell_2$ distance $d_L(t)$ between
the ascending eigenvalue vectors of consecutive normalized Laplacians of the undirected
projection~\citep{chung1997spectral}, and the DeltaCon similarity between consecutive
graphs~\citep{koutra2013deltacon}.

\subsubsection{Stopping Rule and Attractor Classification}
\label{sec:stopping}

A run stops at
\[
t_{\mathrm{conv}}
=\min\{t:\text{exact recurrence or }C^{\mathrm{out}}\text{ plateau}\}.
\]
Exact recurrence requires $E(t)=E(t-p)$ to hold persistently for some
$p\in[1,P_{\max}]$. The plateau condition requires, for
\texttt{patience} consecutive steps (parameters in Table~\ref{tab:stopping}),
\[
\frac{
\left|
\operatorname{mean}(C^{\mathrm{out}}[t-W:t])
-\operatorname{mean}(C^{\mathrm{out}}[t-2W:t-W])
\right|}
{\operatorname{mean}(C^{\mathrm{out}}[t-2W:t-W])}
<\varepsilon_C.
\]

Runs are classified as \texttt{fixed\_point} when the edge set becomes
stationary, \texttt{limit\_cycle}:$p$ when it recurs with
$2\le p\le P_{\max}$, \texttt{plateau} when the retention rate stabilizes
without a detected period, and \texttt{none} when no condition is met within the
simulation horizon. Section~\ref{sec:thm} relates these classes to
Theorem~\ref{thm:equilibrium}.

\begin{table}[t]
\centering
\caption{Stopping-rule parameters.}
\label{tab:stopping}
\begin{tabular}{lll}
\toprule
Parameter & Value & Role \\
\midrule
$W$                   & 10   & averaging-window length \\
$\varepsilon_C$       & 0.01 & maximum relative change between windows \\
\texttt{patience}     & 5    & consecutive steps satisfying the condition \\
\texttt{post\_window} & 10   & steps recorded after detection \\
$P_{\max}$            & 20   & longest tested recurrence period \\
\bottomrule
\end{tabular}
\end{table}

\subsection{Evaluation Metrics}
\label{sec:metrics}

\paragraph{Price of Anarchy.}
The stepwise PoA follows Eq.~\eqref{eq:poa-operational}, with the optimum of
Lemma~\ref{lem:opt} and the decomposition of Remark~\ref{rem:decomp}.

\paragraph{Polarization ($P_z$).}
Following~\citet{chitra2020filter}, polarization is the variance of expressed
opinions:
\begin{equation}
P_z^{(t)}
=\frac{1}{N}\sum_{i=1}^{N}
\left(z_i^{(t)}-\bar{z}^{(t)}\right)^2,
\qquad
\bar{z}^{(t)}=\frac{1}{N}\sum_i z_i^{(t)}.
\end{equation}

\paragraph{Modularity ($Q$ and $Q_{\mathrm{norm}}$).}
Louvain community detection~\citep{blondel2008louvain} is applied to the
undirected projection of $G^{(t)}$, yielding modularity
$Q_{\mathrm{obs}}(t)$~\citep{newman2004finding}. Because random graphs can
exhibit nonzero modularity from degree fluctuations alone
\citep{guimera2004modularity}, we compare the observed graph with a
degree-preserving null model:
\begin{equation}
Q_{\mathrm{norm}}(t)
=Q_{\mathrm{obs}}(t)-\bar{Q}_{\mathrm{rand}}(t),
\qquad
z_Q(t)
=\frac{Q_{\mathrm{obs}}(t)-\bar{Q}_{\mathrm{rand}}(t)}
       {\sigma_{\mathrm{rand}}(t)}.
\end{equation}
The null mean $\bar{Q}_{\mathrm{rand}}$ and standard deviation $\sigma_{\mathrm{rand}}$ are
taken over 20 degree-preserving double-edge-swap randomizations of the undirected
projection and are recomputed at every step. $Q_{\mathrm{norm}}$ is the modularity
attributable to opinion-based sorting beyond the degree sequence.

\subsection{Cross-Run Comparability}
\label{sec:comparability}

Because run lengths differ, outcome metrics are compared at
$t_{\mathrm{conv}}+\texttt{post\_window}$, each run's own evaluation point.

\section{Experiments and Numerical Results}
\label{sec:experiments}
\subsection{Dataset Overview}
\label{sec:dataset}
 
We source data from Pushshift's April 2019 comment
dump~\citep{baumgartner2020pushshift} (RC\_2019-04, 138.5 million comments), of which
1.76 million belong to the \texttt{r/politics} subreddit. The first 500{,}000 valid
\texttt{r/politics} comments in stream order, excluding deleted or removed comments, bodies of
at most 20 characters and comments by deleted authors, enter the clustering stage.
Semantic clustering (Section~\ref{sec:corpus}) identifies 1{,}184 topic clusters across
politics; we focus on two politically contested topics relevant to our simulation:
\textit{gun control} ($n = 3{,}581$ comments) and \textit{abortion}
($n = 1{,}618$ comments).
 
Figure~\ref{fig:stance_dist} shows the stance distributions for both topics under the
regressor of Section~\ref{sec:stance}. Gun control exhibits a
near-neutral mean with high variance, consistent with a bimodal and divided public;
abortion displays a more concentrated opposing consensus. The resulting initial-stance
split among the $N = 50$ sampled agents is 54\%/46\% (oppose/support) for gun control and 68\%/32\% for abortion,
an asymmetry that Section~\ref{sec:res-rq1} shows to be consequential.
 
\begin{figure}[ht]
  \centering
  \includegraphics[width=\linewidth]{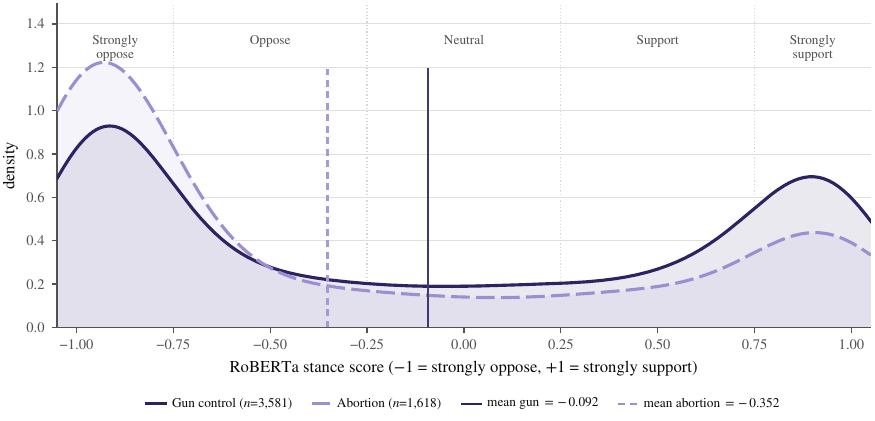}
  \caption{Stance score distributions for gun control ($n=3{,}581$) and abortion
    ($n=1{,}618$) comments in \texttt{r/politics} (April 2019).}
  \label{fig:stance_dist}
\end{figure}
 
\begingroup
\paragraph{The two topics are not independent replicates.}
The 50 sampled agents share identical personas and stubbornness coefficients $\rho_i$
across the two topics; only the initial stance differs.
\endgroup
 
\subsection{Simulation Configuration}
\label{sec:config}
 
All simulations use $N = 50$ agents and $K = 5$ neighbors per agent. The mixing
parameter is varied across
$\alpha \in \{0,\, 0.125,\, 0.25,\, 0.375,\, 0.5,\, 0.625,\, 0.75,\, 0.875,\, 1\}$,
covering the full spectrum from a purely LLM-driven population ($\alpha = 0$) to a
purely cost-minimizing population ($\alpha = 1$).
 
\begingroup
\paragraph{Topology parameters.}
The three initial-topology families of Section~\ref{sec:config-net} are matched on mean
degree (3.84 for BA, 3.84 for ER, 4.00 for WS), so that any
difference between them is attributable to structure rather than to edge count. Their
clustering coefficients remain clearly distinct at 0.179, 0.079, and 0.390
respectively, preserving the contrast of interest.
 
\paragraph{Grid completeness.}
The design yields $n = 60$ runs at every $\alpha$ level and $n = 30$ at every
$(\alpha, \text{topic})$ cell, with no missing cells: all 540/540 runs complete and
enter the analysis. The run is the unit of analysis, so statistical power comes from the
number of conditions and seeds.
\endgroup
 
\subsection{Statistical Analysis}
\label{sec:stats}
 
\begingroup
Variances are markedly unequal across $\alpha$, with the standard deviation of PoA
falling by a factor of about 60 between $\alpha = 0$ and $\alpha = 1$, so comparisons
between $\alpha$ levels use Welch's $t$-test~\citep{welch1947generalization} with
Benjamini--Hochberg FDR control~\citep{benjamini1995controlling} over the family of
adjacent-$\alpha$ comparisons, and we report $q$-values for this family; the two per-topic
contrasts in Section~\ref{sec:res-rq1} are unadjusted $p$-values.
 
Where the claim of interest is equivalence rather than difference, specifically for the
comparison among initial topologies, we test for it directly with two one-sided
tests (TOST)~\citep{schuirmann1987comparison} against a margin fixed in advance at
$\delta = 0.2$ PoA units.
 
All confidence intervals are reported at the 95\% level.
\endgroup
 
\subsection{Implementation Details}
\label{sec:impl}
 
Type-L agents are powered by Phi-4 (14.7B parameters)~\citep{abdin2024phi4} under 4-bit
AWQ quantization~\citep{lin2024awq}, served via vLLM~\citep{kwon2023vllm}; Type-C agents are
closed-form FJ updates.
The full simulation codebase is available at
\url{https://github.com/Evan-Jiamg/Echo-Chamber-Simulation}.

\begingroup
 
\subsection{Numerical Results}
\label{sec:results}

The answers to the three questions of the introduction are summarized below;
Section~\ref{sec:res-conv} first establishes the convergence on which every comparison
rests.

\begin{enumerate}
  \item The efficiency gap is about fivefold, and language agents are displaced from their
    own positions as well as disagreeing with their neighbors (Section~\ref{sec:res-rq1}).
    PoA falls from 5.558 in the purely language population to 1.139 in the purely numerical
    one (Table~\ref{tab:poa-alpha}). From the purely numerical to the purely language
    population the conformity term grows by a factor of 13.2 and the disagreement term by
    3.3, and the conformity share of social cost rises from 16.1\% to 43.6\%
    (Table~\ref{tab:decomp}).
  \item Polarization is lowest in the low-$\alpha$ region, where efficiency is worst
    (Section~\ref{sec:res-rq2}). $P_z$ and $Q_{\mathrm{norm}}$ rise with $\alpha$, while PoA
    holds a statistical plateau up to $\alpha=0.375$ and then falls
    (Tables~\ref{tab:main} and~\ref{tab:poa-alpha}).
  \item Echo chambers come from the rewiring rule rather than the initial topology
    (Section~\ref{sec:res-topology}). They form under all three initial topologies, and at
    $\alpha=1$ the three are equivalent in PoA (Table~\ref{tab:topology}).
\end{enumerate}

\begin{table}[!htbp]
\centering
\caption{Metrics at each run's converged state ($t_{\mathrm{conv}} + 10$), by topic and
mixing ratio. $n = 30$ per row (3 topologies $\times$ 10 seeds), mean $\pm$ 95\% CI.
$\downarrow$ denotes lower-is-better; $Q_{\mathrm{norm}}$ and $z_Q$ as in
Section~\ref{sec:metrics}. Bold marks the lowest PoA for each topic.}
\label{tab:main}
\begin{tabular}{llcccc}
\toprule
Topic & $\alpha$ & $P_z$ & $Q_{\mathrm{norm}}$ & $z_Q$ & PoA $\downarrow$ \\
\midrule
\multirow{9}{*}{Gun Control}
 & 0.0   & $0.261 \pm 0.013$ & $0.241 \pm 0.009$ & $22.8 \pm 2.1$ & $5.681 \pm 0.340$ \\
 & 0.125 & $0.312 \pm 0.016$ & $0.250 \pm 0.009$ & $23.0 \pm 1.5$ & $5.480 \pm 0.312$ \\
 & 0.25  & $0.364 \pm 0.013$ & $0.259 \pm 0.009$ & $23.8 \pm 1.6$ & $5.642 \pm 0.396$ \\
 & 0.375 & $0.409 \pm 0.014$ & $0.263 \pm 0.009$ & $22.5 \pm 1.3$ & $5.477 \pm 0.325$ \\
 & 0.5   & $0.451 \pm 0.012$ & $0.277 \pm 0.007$ & $23.3 \pm 1.3$ & $5.374 \pm 0.395$ \\
 & 0.625 & $0.494 \pm 0.011$ & $0.286 \pm 0.010$ & $23.1 \pm 1.3$ & $4.424 \pm 0.344$ \\
 & 0.75  & $0.538 \pm 0.008$ & $0.314 \pm 0.008$ & $24.9 \pm 1.9$ & $3.821 \pm 0.274$ \\
 & 0.875 & $0.582 \pm 0.006$ & $0.350 \pm 0.009$ & $24.5 \pm 1.3$ & $3.220 \pm 0.222$ \\
 & 1.0   & $0.619 \pm 0.000$ & $0.427 \pm 0.002$ & $27.4 \pm 1.9$ & $\mathbf{1.149 \pm 0.008}$ \\
\midrule
\multirow{9}{*}{Abortion}
 & 0.0   & $0.347 \pm 0.017$ & $0.194 \pm 0.006$ & $18.4 \pm 1.4$ & $5.436 \pm 0.535$ \\
 & 0.125 & $0.364 \pm 0.011$ & $0.199 \pm 0.007$ & $20.7 \pm 1.6$ & $6.322 \pm 0.389$ \\
 & 0.25  & $0.426 \pm 0.015$ & $0.191 \pm 0.007$ & $18.8 \pm 1.6$ & $6.077 \pm 0.367$ \\
 & 0.375 & $0.463 \pm 0.016$ & $0.241 \pm 0.013$ & $23.2 \pm 1.9$ & $5.815 \pm 0.385$ \\
 & 0.5   & $0.480 \pm 0.015$ & $0.283 \pm 0.010$ & $24.4 \pm 2.3$ & $5.117 \pm 0.298$ \\
 & 0.625 & $0.510 \pm 0.013$ & $0.297 \pm 0.008$ & $24.0 \pm 1.9$ & $4.670 \pm 0.296$ \\
 & 0.75  & $0.534 \pm 0.010$ & $0.317 \pm 0.008$ & $24.4 \pm 1.5$ & $3.853 \pm 0.288$ \\
 & 0.875 & $0.554 \pm 0.008$ & $0.343 \pm 0.006$ & $25.0 \pm 1.4$ & $2.758 \pm 0.306$ \\
 & 1.0   & $0.575 \pm 0.000$ & $0.419 \pm 0.002$ & $25.8 \pm 1.4$ & $\mathbf{1.129 \pm 0.001}$ \\
\bottomrule
\end{tabular}
\end{table}
 
\begin{figure}[!htbp]
  \centering
  \includegraphics[width=\linewidth]{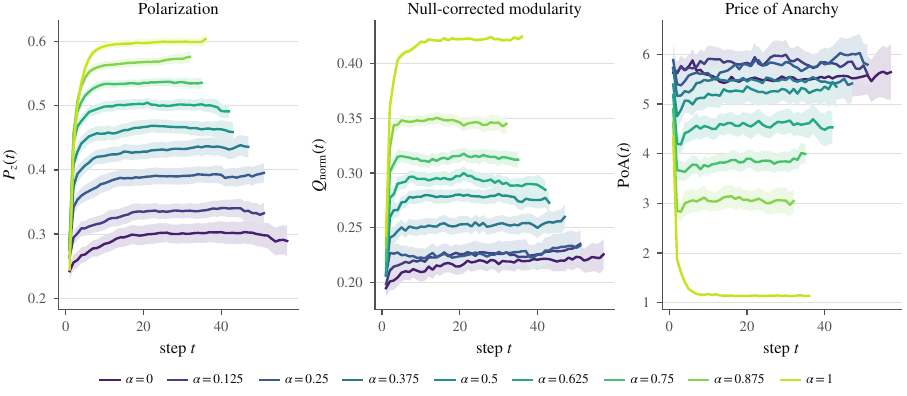}
  \caption{Per-step trajectories of $P_z$, $Q_{\mathrm{norm}}$ and PoA by mixing ratio (mean
    with 95\% CI band). All three are flat well
    before the graph enters its attractor, so their values in Table~\ref{tab:main} do not
    depend on the precise step at which they are read.}
  \label{fig:trajectories}
\end{figure}
 
Table~\ref{tab:main} is the primary result, and Figure~\ref{fig:trajectories} shows the same
quantities over time.

\subsubsection{Convergence as a Precondition}
\label{sec:res-conv}
 
Theorem~\ref{thm:equilibrium} guarantees that OGA learning reaches an approximate Nash
equilibrium at some iteration or, otherwise, a bounded price of anarchy;
the form of the attractor the coevolving graph enters is an empirical question for
both agent types, because $K$-NN rewiring is a discontinuous map applied on top of the
Friedkin--Johnsen update and Type-L positions have no closed form. That both modes of updating
converge, as shown below, licenses comparing converged states across $\alpha$.
 
\begin{figure}[!htbp]
  \centering
  \begin{subfigure}[t]{0.49\textwidth}
    \centering
    \includegraphics[width=\linewidth]{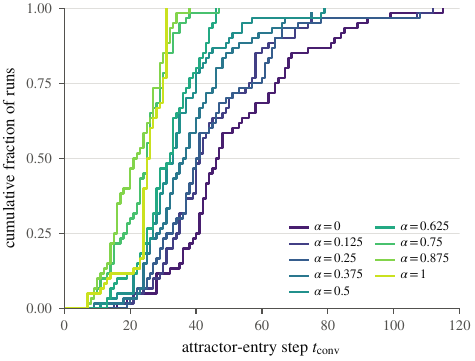}
    \caption{Attractor-entry step $t_{\mathrm{conv}}$, as empirical cumulative
      distribution functions.}
    \label{fig:tconv}
  \end{subfigure}
  \hfill
  \begin{subfigure}[t]{0.49\textwidth}
    \centering
    \includegraphics[width=\linewidth]{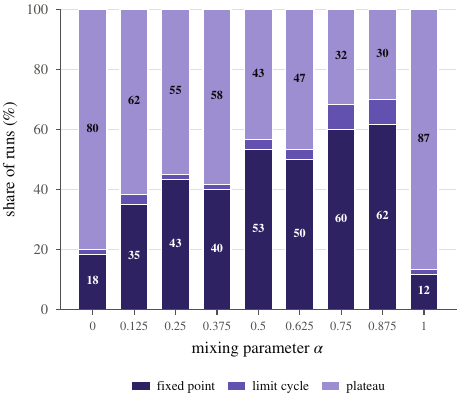}
    \caption{Attractor composition. Exactly stationary edge sets are rare at both pure
      compositions and most common in the mixed regime.}
    \label{fig:attractor}
  \end{subfigure}
  \caption{Settling behavior by mixing ratio ($n = 60$ per curve and per bar, pooled over
    three topologies and two topics).}
  \label{fig:settling}
\end{figure}
 
\paragraph{Both pure populations converge, and mostly as a plateau.}
Figure~\ref{fig:settling} summarizes settling behavior across the grid. Neither pure
population settles predominantly to an exact fixed point: both are classified
\texttt{plateau} in the large majority of runs, and the purely Type-L population reaches
an exactly stationary edge set slightly more often than the purely Type-C one, whose
opinion dynamics converge analytically. Under $K$-NN rewiring, convergence
takes the form of a plateau in the rate of structural change rather than a frozen graph,
and Type-L agents attain that state at least as reliably as the closed-form Type-C
dynamics.
 
\paragraph{Settling time is a nonmonotone function of composition.}
Figure~\ref{fig:tconv} shows $t_{\mathrm{conv}}$ falling monotonically in $\alpha$ from
52.2 steps at $\alpha = 0$ to 21.7 at $\alpha = 0.875$, and then rising to 25.1 at
$\alpha = 1.0$. Run lengths are right-skewed; at $\alpha = 0$ the
median is 46 steps and the mean is 52.2.
 
\paragraph{Attractor composition shows the same reversal.}
Figure~\ref{fig:attractor} resolves the composition by class. Exactly stationary edge
sets account for 18\% of runs at $\alpha = 0$ and 12\% at $\alpha = 1$, whereas
mixtures reach 62\% at $\alpha = 0.875$: a population containing
six LLM agents attains an exact fixed point five times as often as one containing none.

\begin{figure}[!htbp]
  \centering
  \begin{subfigure}[t]{0.49\textwidth}
    \centering
    \includegraphics[width=\linewidth]{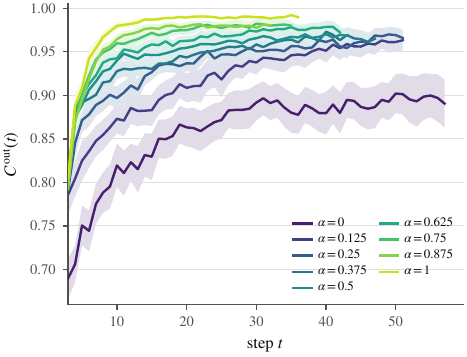}
    \caption{Temporal correlation coefficient $C^{\mathrm{out}}(t)$, the fraction of
      out-neighbors retained between consecutive steps.}
    \label{fig:cout}
  \end{subfigure}
  \hfill
  \begin{subfigure}[t]{0.49\textwidth}
    \centering
    \includegraphics[width=\linewidth]{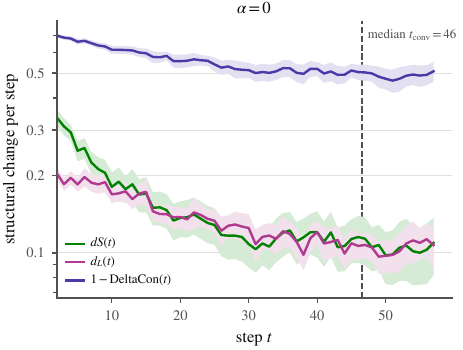}
    \caption{Three measures at $\alpha = 0$ on a logarithmic vertical axis, all
      flattening ahead of the median $t_{\mathrm{conv}}$ (dashed).}
    \label{fig:agreement}
  \end{subfigure}
  \caption{Structural change per step, defined from $t = 2$ (mean with 95\% CI band).
    Panel~(a) is the statistic the stopping rule reads; panel~(b) shows three measures
    that take no part in it.}
  \label{fig:churn}
\end{figure}
 
\paragraph{Rewiring churn separates the purely Type-L population.}
Figure~\ref{fig:churn} traces the per-step structural change, and
Figure~\ref{fig:cout} the stopping statistic itself. The plateau height of
$C^{\mathrm{out}}$ is $0.913$ (SD $0.069$) at $\alpha = 0$ against $0.990$ (SD $0.004$) at
$\alpha = 1$, and elevated churn is confined to $\alpha = 0$. Uniformly random rewiring
gives $C^{\mathrm{out}}\approx K/(N-1)\approx0.10$, so both plateau heights reflect
persistent neighborhoods. The Type-C population comes to rest close to a frozen structure,
whereas the Type-L
population reaches a state in which the rate of turnover has stabilized while a small
fraction of ties continues to change.
 
\paragraph{Measures outside the stopping rule corroborate it.}
Convergence is declared on $C^{\mathrm{out}}$ alone. Figure~\ref{fig:agreement} shows three
measures that take no part in that decision (the normalized Hamming distance between
consecutive edge sets, the spectral distance between their Laplacians, and
$1 - \mathrm{DeltaCon}$) at $\alpha = 0$, the slowest-settling composition: each flattens at
around $t = 25$ to $30$, ahead of the median $t_{\mathrm{conv}}$ of 46, so the measure the
rule reads is the last of four to register the settled regime. The criterion is satisfied
within the simulation horizon in all runs (540/540), and no run is excluded.

\subsubsection{The Efficiency Gap and Its Decomposition}
\label{sec:res-rq1}
 
\paragraph{The purely Type-C baseline.}
At $\alpha = 1$ the measured PoA is $1.139 \pm 0.005$, the inefficiency of the opinion
formation game itself, against which every other composition is compared. On the
$\alpha = 1$ runs with a stationary graph the computed opinions coincide with the fixed point
of Lemma~\ref{lem:fixed}. The value is consistent with the
$9/8$ bound (Remark~\ref{rem:98}), and it lies well below the worst case of the $K$-NN game
for $\rho<1$, which is at least $1/\rho^2\ge2$ over the range of $\rho_i$ used
here~\citep{bhawalkar2013}.
 
\begin{table}[!htbp]
\centering
\begin{minipage}[t]{0.44\textwidth}
  \centering
  \caption{Price of Anarchy by mixing ratio, pooled ($n = 60$ per row), mean $\pm$ 95\% CI. The $q$ column
    gives the Welch test against the preceding level under Benjamini--Hochberg control;
    bold marks $q < 0.05$.}
  \label{tab:poa-alpha}
  \begin{tabular}{lcl}
  \toprule
  $\alpha$ & PoA & $q$ vs.\ prev. \\
  \midrule
  0.0   & $5.558 \pm 0.309$ & --- \\
  0.125 & $5.901 \pm 0.265$ & 0.127 \\
  0.25  & $5.860 \pm 0.268$ & 0.826 \\
  0.375 & $5.646 \pm 0.249$ & 0.280 \\
  0.5   & $5.246 \pm 0.242$ & \textbf{0.036} \\
  0.625 & $4.547 \pm 0.222$ & $\mathbf{8.5{\times}10^{-5}}$ \\
  0.75  & $3.837 \pm 0.193$ & $\mathbf{1.1{\times}10^{-5}}$ \\
  0.875 & $2.989 \pm 0.193$ & $\mathbf{3.0{\times}10^{-8}}$ \\
  1.0   & $1.139 \pm 0.005$ & $\mathbf{3.9{\times}10^{-26}}$ \\
  \bottomrule
  \end{tabular}
\end{minipage}
\hfill
\begin{minipage}[t]{0.50\textwidth}
  \centering
  \caption{The same social cost separated into its two components, pooled
    ($n = 60$ per row). Their sum is the Price of Anarchy at left; bold marks the largest and smallest
    conformity share.}
  \label{tab:decomp}
  \begin{tabular}{lccc}
  \toprule
  $\alpha$ & Disagreement & Conformity & Conf.\ share \\
  \midrule
  0.0   & 3.136 & 2.422 & \textbf{43.6\%} \\
  0.125 & 3.690 & 2.211 & 37.5\% \\
  0.25  & 3.882 & 1.978 & 33.7\% \\
  0.375 & 3.822 & 1.824 & 32.3\% \\
  0.5   & 3.610 & 1.636 & 31.2\% \\
  0.625 & 3.127 & 1.420 & 31.2\% \\
  0.75  & 2.678 & 1.159 & 30.2\% \\
  0.875 & 2.175 & 0.814 & 27.2\% \\
  1.0   & 0.956 & 0.183 & \textbf{16.1\%} \\
  \bottomrule
  \end{tabular}
\end{minipage}
\end{table}
 
The fall in PoA is not a smooth interpolation (Table~\ref{tab:poa-alpha}): no comparison up
to and including the one reaching $\alpha = 0.375$ is
significant ($q > 0.12$), and every comparison from $\alpha = 0.5$ onward is
significant ($q < 0.04$). The low-$\alpha$ region is a
statistical plateau, and the decline begins only once cost-minimizing agents approach
parity. PoA is minimized at $\alpha = 1$ and no interior mixture approaches it; since
polarization is lowest in the low-$\alpha$ plateau, where efficiency is worst, no composition
optimizes both.
 
\begin{figure}[!htbp]
  \centering
  \includegraphics[width=0.75\linewidth]{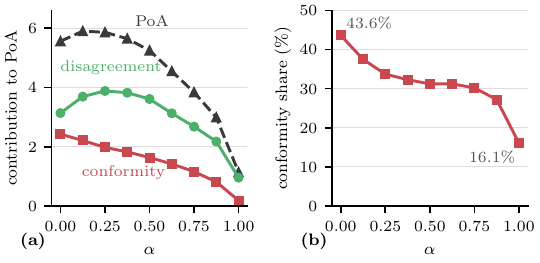}
  \caption{Social-cost decomposition against the mixing ratio, on the values of
    Table~\ref{tab:decomp}. (a) The disagreement term rises to a peak at $\alpha = 0.25$
    before falling, while the conformity term declines throughout; each component is
    divided by the planner optimum of the same step, so the two sum to the Price of
    Anarchy, shown dashed. (b) The conformity share falls from 43.6\% to 16.1\%.}
  \label{fig:decomp}
\end{figure}

Table~\ref{tab:decomp} and Figure~\ref{fig:decomp} separate the two components.
Between $\alpha = 1$ and $\alpha = 0$ the disagreement term grows by a factor of
3.3 ($0.956 \to 3.136$) while the conformity term grows by a factor of 13.2
($0.183 \to 2.422$), and the conformity share falls as $\alpha$ rises. The two components
contribute about equally to the gap of 4.42 between the pure populations, 2.18 from
disagreement and 2.24 from conformity.
 
Within the plateau the two topics differ: from $\alpha = 0$ to $\alpha = 0.125$ PoA rises
from $5.436$ to $6.322$ for abortion ($p = 0.008$), while gun control is flat ($p = 0.38$) and
attains its maximum at $\alpha = 0$. Because the two topics share agent
personas and stubbornness coefficients (Section~\ref{sec:dataset}), the difference is
attributable to the stance distribution rather than to the agents themselves.
 
\subsubsection{Polarization and Echo Chambers}
\label{sec:res-rq2}

\begin{figure}[!htbp]
  \centering
  \includegraphics[width=0.75\linewidth]{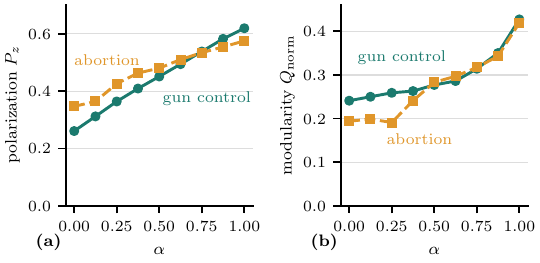}
  \caption{Polarization and echo-chamber structure against the mixing ratio, by topic.
    (a) The variance of expressed opinions rises with $\alpha$ for both topics. (b) The
    null-corrected modularity rises over the same range. Both vertical axes start at zero.}
  \label{fig:structure}
\end{figure}
 
Figure~\ref{fig:structure} shows both measures by topic. Pooled,
$P_z$ rises from $0.304 \pm 0.015$ to $0.597 \pm 0.006$ and $Q_{\mathrm{norm}}$ from
$0.218 \pm 0.008$ to $0.423 \pm 0.002$ across the range of $\alpha$. Echo-chamber
structure forms under every composition, but its strength depends on agent type:
Type-L populations produce measurably weaker polarization and weaker community
structure.
 
All claims use $Q_{\mathrm{norm}}$ rather than bare $Q$. At $\alpha = 1$ the bare value
for gun control is $0.763$ while the degree-preserving null attains $0.336$ on the same degree sequence,
so roughly 44\% of raw modularity follows from the constant out-degree the $K$-NN rule
imposes. What remains is far from null at every composition: $z_Q$ places the observed
modularity between 18 and 27 standard deviations above the null distribution across the
grid.
 
\subsubsection{The Role of the Rewiring Mechanism}
\label{sec:res-topology}
 
\begin{table}[!htbp]
\centering
\caption{PoA by initial topology ($n = 20$ per cell), mean $\pm$ 95\% CI, with the largest
absolute pairwise gap between families at each level and its TOST equivalence $p$-value
against $\delta = 0.2$; bold marks established equivalence.}
\label{tab:topology}
\begin{tabular}{lccccc}
\toprule
$\alpha$ & BA & ER & WS & largest gap & TOST $p$ \\
\midrule
0.0  & $5.348 \pm 0.503$ & $5.766 \pm 0.737$ & $5.562 \pm 0.395$ & $0.418$ & 0.694 \\
0.25 & $5.944 \pm 0.356$ & $5.961 \pm 0.531$ & $5.673 \pm 0.557$ & $0.288$ & 0.594 \\
0.5  & $5.078 \pm 0.373$ & $5.422 \pm 0.417$ & $5.237 \pm 0.521$ & $0.344$ & 0.704 \\
0.75 & $3.912 \pm 0.420$ & $4.051 \pm 0.332$ & $3.549 \pm 0.247$ & $0.501$ & 0.932 \\
1.0  & $1.138 \pm 0.008$ & $1.134 \pm 0.007$ & $1.144 \pm 0.010$ & $0.010$ & $\mathbf{<0.0001}$ \\
\bottomrule
\end{tabular}
\end{table}
 
Table~\ref{tab:topology} reports PoA by family. Across all nine $\alpha$ levels and both
PoA and $P_z$, the confidence intervals of the three topology families overlap and no
comparison attains significance. The two one-sided tests (TOST) against the margin $\delta = 0.2$
(Section~\ref{sec:stats}) establish equivalence at $\alpha = 1$ ($p < 0.0001$). At the
five levels of Table~\ref{tab:topology} the initial topology shifts PoA by at most 0.50,
against a composition effect of 4.42.
 
The three families are matched on mean degree while their clustering coefficients differ by
a factor of five (Section~\ref{sec:config}), and $K$-NN rewiring overwrites the initial edge set
at the first step and thereafter rebuilds every out-neighborhood from the current expressed
opinions. That no difference is detectable under this contrast, at $\alpha = 0$ as much as at
$\alpha = 1$, indicates that echo-chamber formation here requires no pre-existing hub
structure and is not a property of the FJ closed form: echo chambers follow from the
connection rule rather than from the way the population was wired at the outset, and the
composition of the population determines how deep the clustering is, not whether it occurs.
 
\endgroup

\FloatBarrier
\section{Discussion, Conclusions and Future Work}
\label{sec:discussion}
        \subsection{Interpretation of Key Findings}

    \begingroup
    \paragraph{Language reasoning does not anchor an agent to its original position.}
    Both kinds of member carry their original position into every round; they differ in how
    strongly that position binds them. For an analytical member the original position is
    written into the cost function and acts as a restoring force: however uniform its
    neighbors become, departing from it costs the same, so an analytical member compromises
    with its neighbors without drifting far (Lemma~\ref{lem:br}).

    The two kinds also engage their surroundings at different depths.
    From each neighbor an analytical member receives a single number, and its best response
    (Section~\ref{sec:framework-updates}) is a weighted average of those numbers and its own
    stance: it knows where a neighbor stands but not why. A language member receives the
    reasons its neighbors have written and answers in its own words, facing the arguments
    behind a position, not only the position. After the first exchange the opinions of both
    kinds crowd into the middle, with $P_z$ of only about 0.25 (Figure~\ref{fig:trajectories}),
    and rewiring then brings to each member the peers whose expressed opinions lie closest to
    its original stance. The cost function pulls an analytical member back; a language member
    restates its position in its own words and returns only part of the way. The stance in its
    prompt never changes, but the stance it voices is rewritten in every exchange.
    \endgroup

\subsection{Conclusions}

    \begingroup
    The Price of Anarchy measures how far a settled community lies from the best state it
    could have reached. Theorem~\ref{thm:equilibrium} guarantees that OGA learning
    reaches an approximate equilibrium or, otherwise, a bounded price of anarchy, and a community reasoning in language settles as reliably as one
    updating by a cost function. What we compare is therefore the final shape of two
    communities rather than snapshots taken at arbitrary times.

    The cost a community pays comes from quarrels among its members and also from members
    quietly leaving what they originally believed. Members that update by a cost function hold
    to their positions and compromise with those around them without drifting far. Members
    that reason in language voice new positions over repeated exchanges of reasons; they leave
    their original place and still do not stand with those around them. Half of their extra
    cost falls here, and each member bears it privately. Such a community looks less divided,
    but mildness has a price. When the opinions members voice are more concentrated than their
    original positions, the cost of leaving those positions has an unavoidable floor, and the
    floor rises with the concentration (Lemma~\ref{lem:floor}).

    Recommendation algorithms on social media favor controversial topics, because divisive,
    hostile content draws more engagement~\citep{rathje2021outgroup,milli2025engagement}. In
    that environment it matters increasingly whether a community merely looks mild or has
    actually reached consensus. A forum that looks calm may simply be one whose members no
    longer voice what they first thought. A forum split into camps need not be one of stubborn
    members either: in our model the camps come from members' preference for associating with
    like-minded others~\citep{mcpherson2001birds} rather than from the network they started in.
    An audit of a platform should ask how widely opinions are spread, and also how far users now
    stand from what they first believed.
    \endgroup

    \subsection{Limitations and Future Work}
     
    \begingroup
    \begin{enumerate}
    \item All Type-L agents here are instances of a single model, Phi-4, under one prompt
    framework, with a distinct persona for each agent. Running the same grid across model families,
    model sizes and prompt frameworks would separate a property of language-based opinion updating
    from a property of one instance.

    \item Efficiency is stated relative to the Friedkin--Johnsen benchmark. A human baseline on the same corpus,
    in which participants revise their stance under the same rewiring rule, would state it
    relative to human behavior instead.
    \end{enumerate}

    We plan to run both experiments. Frameworks such as AgentSociety~\citep{piao2025agentsociety} and
    OASIS~\citep{yang2024oasis} already run LLM-agent societies of $10^4$ to $10^6$ agents;
    along these lines, we plan to scale the game from $N=50$ to $10^3$ to $10^4$ agents.
    Finally, the Pushshift archives used here predate the 2023 restrictions on Reddit data
    access~\citep{poudel2024postapi}; we intend to initialize runs from current platform
    data, such as Bluesky interaction networks collected through its open
    API~\citep{quelle2025bluesky} and stance-labelled political corpora such as
    PolitiSky24~\citep{rostami2025politisky24}.

    \endgroup

\bibliographystyle{plainnat}
\bibliography{references}

\end{document}